\documentclass[%
 reprint,
superscriptaddress,
 amsmath,amssymb,
 aps,
prl,
]{revtex4-2}

\usepackage{graphicx}
\usepackage{dcolumn}
\usepackage{bm}
\usepackage{hyperref}

\usepackage{amsmath}
\usepackage{amssymb}
\usepackage{stackrel}
\usepackage{tikz-cd}
\usepackage{tikz}
\usepackage[cm]{fullpage}
\usepackage{amsthm}
\usepackage{mathtools}
\usepackage{enumitem}
\usepackage{bbold}
\usepackage[normalem]{ulem}
\usepackage{hyperref}
\usepackage{cleveref}
\usepackage{appendix}
\usepackage{comment}
\usepackage{subcaption}
\usepackage{centernot}
\usepackage{braket}

\newtheorem{Theorem}{Theorem}
\newtheorem{Lemma}{Lemma}
\newtheorem{Proposition}{Proposition}

\theoremstyle{definition}
\newtheorem{Definition}{Definition}

\DeclareMathOperator{\Odd}{Odd}

\renewcommand{\H}{\mathcal{H}}

\begin{document}

\preprint{APS/123-QED}

\title{Parity Quantum Computing as YZ-Plane Measurement-Based Quantum Computing}

\author{Isaac D. Smith}
\email{isaac.smith@uibk.ac.at}
\author{Hendrik Poulsen Nautrup}
\affiliation{University of Innsbruck, Institute for Theoretical Physics}
\author{Hans J. Briegel}
\affiliation{University of Innsbruck, Institute for Theoretical Physics}
\affiliation{Fachbereich Philosophie, Universität Konstanz, Konstanz, Germany}

\date{\today}

\begin{abstract}
We show that universal parity quantum computing employing a recently introduced constant depth decoding procedure is equivalent to measurement-based quantum computation (MBQC) on a bipartite graph using only YZ-plane measurements. We further show that \textit{any} unitary MBQC using only YZ-plane measurements must occur on a bipartite graph. These results have a number of consequences and open new research avenues for both frameworks.
\end{abstract}

\maketitle




In the present era of pre-fault-tolerant quantum computation \cite{Preskill2018quantumcomputingin}, there exists an array of theoretical proposals for computation that display certain advantages and differing levels of suitability for implementation on current physical devices.

Parity quantum computation \cite{lechner2015quantum, fellner_22, Fellner_PRA_22, Fellner2023parityquantum, DriebSchon2023parityquantum, Ender2023parityquantum} refers to one such proposal, initially based on quantum annealing \cite{lechner2015quantum}. The universal parity computing framework \cite{fellner_22} leverages the properties of a certain type of quantum state encoding, the parity encoding. This encoding maps an $n$-qubit logical state onto $n(n+1)/2$ physical qubits, some of which obtain parity information related to subsets of logical qubits. Consequently, certain rotations acting locally on these parity qubits translate to multi-qubit logical rotations on the corresponding subset \cite{fellner_22}. The parity code is in particular a stabiliser code \cite{gottesman1997stabilizer,nielsen_chuang_2010} and many of the properties of the code are well understood using the stabiliser formalism.

Stabiliser states and stabiliser codes are known to have a canonical form, namely graph states \cite{Hein_04,hein2006ent} and graph codes \cite{schlingemann2002logical,schlingemann2001stabilizer,Schlingemann_01} respectively. Graph states form an important class of highly entangled states that support measurement-based quantum computation (MBQC) \cite{raussendorf2001one,briegel2009measurement,raussendorf2001computational,raussendorf2003measurement,walther2005experimental,nautrup2023measurement}. MBQC is a well-known alternative to the quantum circuit model driven by single qubit projective measurements instead of unitary gates.

Recently, a proposal for measurement-based encoding and decoding procedures were put forward for the parity computing regime \cite{messinger2023constant}, demonstrating beneficial properties in terms of their computational depth. Due to the close connection between stabiliser codes and graph codes, an investigation of the potential connections to MBQC is warranted, which we initiate in this work. 

After presenting the required background, we demonstrate that every parity code is local Clifford equivalent to a bipartite graph code (\Cref{prop:parity_is_bipartite}). Consequently, we show that parity quantum computation with the measurement-based decoding is MBQC where all measurements are from the YZ-plane of the Bloch sphere, and where re-entanglement and some local operations are allowed (\Cref{thm:parity_is_MBQC}). We further show that any MBQC using only YZ-plane measurements and with input and output sets of equal size must use bipartite graph states (\Cref{thm:MBQC_YZ_I_O}). To conclude, we briefly outline some consequences of these results for both computing paradigms.


\section{Background} \label{sec:background}

\textit{Parity Quantum Computing - } A parity quantum computation commences by encoding the computational input state using the LHZ architecture \cite{lechner2015quantum} (see \Cref{fig:LHZ_example}). The computation proceeds by applying unitaries from a native gate set for the parity encoding, such as that outlined in \cite{fellner_22}, which largely consists of local rotations. To finish, a decoding procedure returns the computational output. Presently, we will focus on the parity encoding procedure and universal gate set presented in \cite{fellner_22} in combination with the measurement-based decoding procedure outlined in \cite{messinger2023constant}.

The parity encoding procedure maps a state on $n$ logical qubits to a state on $n(n+1)/2$ physical qubits. Following \cite{fellner_22}, we consider an LHZ layout where $n$ physical qubits (the `data' qubits) directly correspond to the $n$ logical qubits. The remaining $N = n(n-1)/2$ qubits will be referred to as `parity' qubits. We denote the sets of data and parity qubits by $I$ and $V\setminus I$ respectively.

Encoding consists of applying a sequence of CNOTs to an input state $\ket{\psi}$ and the parity qubits, which are all initialised to $\ket{0}$. Letting $C$ represent the set of control-target pairs, the encoded state is:
\begin{align}
\ket{LHZ_{\psi}} = U_{\text{enc}}\ket{0}^{\otimes N}\ket{\psi}  = \prod_{(c,t) \in C}\text{CNOT}_{(c,t)}\ket{0}^{\otimes N}\ket{\psi}.
\end{align}
Different constraint sets $C$ can produce the same encoded state. For example, $C$ could contain only pairs where every control is a data qubit and every target a parity qubit, which may involve non-nearest neighbour interactions for a given physical layout. Equivalently, it is possible to take $C$ to contain only nearest-neighbour CNOTs, where now some control qubits are parity qubits (see e.g., \Cref{fig:LHZ_example}). The compilation of a given parity code into a nearest-neighbour layout is an interesting optimisation problem and a topic of ongoing research \cite{Ender2023parityquantum,DriebSchon2023parityquantum,Fellner2023parityquantum}.

For each $\ket{\psi}$, $\ket{LHZ_{\psi}}$ is a state in the parity codespace for the given architecture. The stabilizer of the parity code is generated by operators of the form
\begin{align}
K'_{(ij...k)} := Z_{(ij...k)} \otimes Z_{i} \otimes Z_{j} \otimes ... \otimes Z_{k}
\end{align}
where the single subscripts $i$, $j$, etc. indicate data qubits and the subscript $(ij...k)$ indicates the parity qubit that encodes the parity information of data qubits $i$, $j$ and so on. The operators $K'_{(ij...k)}$ for all parity qubits are mutually independent and generate the codespace. Note that often each parity qubit is taken to encode the parity of just two data qubits.

A benefit of this encoding consists in the ability to implement diagonal multi-qubit logical operations via single qubit physical rotations. For example, applying a local $Z$-rotation to a parity qubit $(ij)$ effectively applies a logical $Z_{i} \otimes Z_{j}$-rotation, from which a controlled-phase gate between logical qubits $i$ and $j$ can be obtained via local $Z$-rotations on the corresponding data qubits \cite{fellner_22}. For full universal quantum computation, in conjunction with $Z$-rotations and controlled-phase gates, it suffices to be able to implement a logical $X$-rotation. For a data qubit $i$, this can be done via a decoding sequence of CNOTs along all parity qubits containing parity information about $i$, a local $X$-rotation at data qubit $i$, and a re-encoding sequence of CNOTs (see \cite{fellner_22} for more details).

Until recently, the typical parity decoding procedures involved applying the encoding sequence of CNOT gates in reverse. In \cite{messinger2023constant}, an equivalent decoding procedure was proposed involving local $X$-measurements on parity qubits followed by local $Z$-operations conditional on measurement outcomes. One benefit of this approach is that full and partial decoding can be performed in constant-depth regardless of the size of the architecture.

For this gate set and measurement-based decoding, a unitary $U$ applied to input state $\ket{\psi}$ in the parity regime can be decomposed into a series of layers, where each layer involves parity qubit rotations followed by decoding. For notational simplicity, we consider full decoding in each layer. Denoting the set of layers by $L$, the set of data qubits by $I$ and the set of parity qubits by $V \setminus I$, the computation can be written as:
\begin{align}
U\ket{\psi} &= \prod_{l =1}^{L} \left(U_{\text{data}}^{(l)}(\boldsymbol{\alpha}^{(l)},\boldsymbol{\phi}^{(l)}) D_{\text{dec}}^{(l)}(\boldsymbol{\theta}^{(l)})  U_{\text{enc}} \ket{0}^{\otimes N}\right) \ket{\psi} \label{eq:parity_comp}
\end{align}
where $U_{\text{enc}}$ acts on both the ancilla $\ket{0}^{\otimes N}$ as well as the qubits in $I$, $D_{\text{dec}}^{(l)}(\boldsymbol{\theta}^{(l)})$ is the operator involving parity qubit rotations and decoding for layer $l$ given by
\begin{align}
D_{\text{dec}}^{(l)}(\boldsymbol{\theta}^{(l)}) := \bigotimes_{(ij...k) \in V\setminus I} \bra{+_{(ij...k)}}R_{Z_{(ij...k)}}(\theta_{(ij...k)}^{(l)})
\end{align}
and $U_{\text{data}}^{(l)}(\boldsymbol{\alpha}^{(l)},\boldsymbol{\phi}^{(l)})$ is the product of local rotation on data qubits for layer l given by:
\begin{align}
U_{\text{data}}^{(l)}(\boldsymbol{\alpha}^{(l)},\boldsymbol{\phi}^{(l)}) := \bigotimes_{i \in I} R_{X_{i}}(\alpha_{i}^{(l)})R_{Z_{i}}(\phi_{i}^{(l)}).
\end{align}
Note that the only distinction in the case of partial decoding is that $D_{\text{dec}}^{(l)}$ contains measurements in some subset of $V \setminus I$, the ancilla prepared in $\ket{0}$ in the subsequent layer correspond to the same subset, and the relevant $U_{\text{enc}}$ is replaced by $U_{\text{enc}}^{(l+1)}$ which applies only the appropriate CNOTs to re-encode back to the full LHZ state. 

\begin{figure*}
\begin{subfigure}{0.45\textwidth}
\centering
\includegraphics[width=\textwidth]{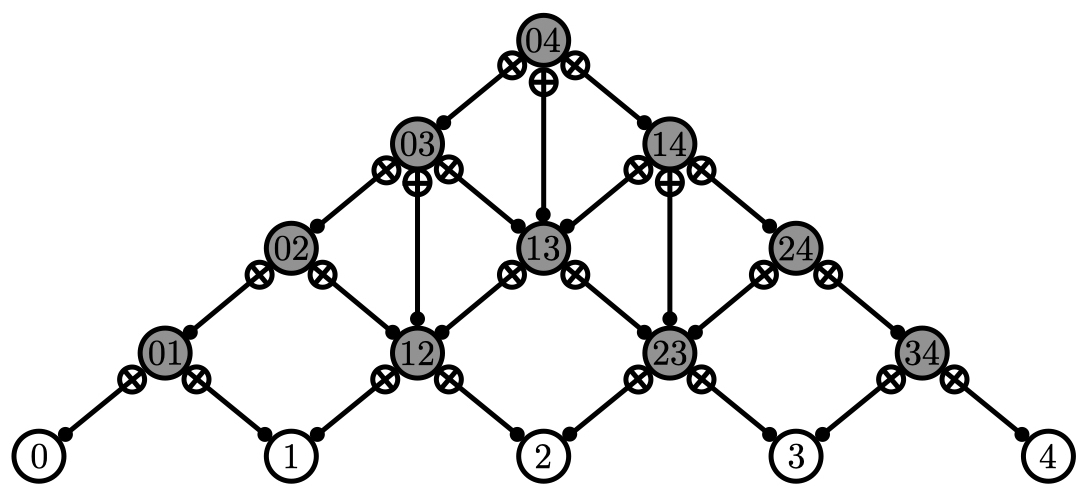}
\caption{An example of an LHZ architecture.}
\label{fig:LHZ_example}
\end{subfigure}
\hfill
\begin{subfigure}{0.45\textwidth}
\centering
\includegraphics[width=\textwidth]{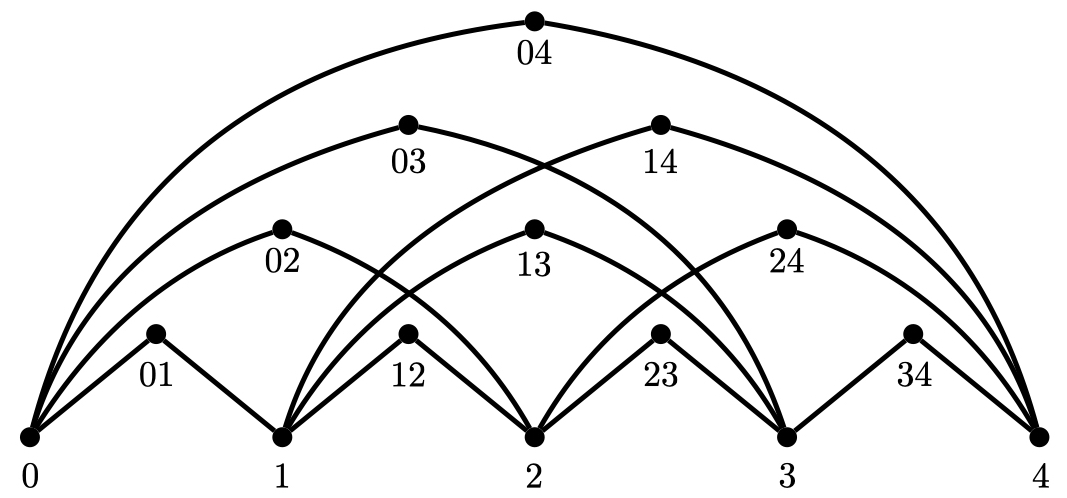}
\caption{The equivalent bipartite graph code.}
\label{fig:graph_eg}
\end{subfigure}
\caption{(a) The parity encoding encodes an input state prepared on the data qubits (white circles). The grey circles denote parity qubits prepared in $\ket{0}$ and CNOT gates are applied to data and parity according to the layout as shown. This parity code in  is equivalent to a graph code for the bipartite graph shown in (b).}
\end{figure*}

\textit{Measurement-Based Quantum Computing - } Measurement-based quantum computing (MBQC) \cite{raussendorf2001one,briegel2009measurement,raussendorf2001computational,raussendorf2003measurement} consists of three things: (i) a highly entangled graph state \cite{Hein_04,hein2006ent}, (ii) a sequence of single qubit projective measurements in certain planes of the Bloch sphere, and (iii) classical, adaptive corrections of future measurements conditioned on prior measurement outcomes. Despite the indeterminacy of quantum measurements, deterministic computation can be performed, provided the sequence of measurements and underlying graph state satisfy certain properties \cite{browne2007generalized}.

Graph states take their name from their connection to mathematical graphs, where vertices correspond to qubits and edges correspond to two-qubit gates. We consider here graph states where a computational input state $\ket{\psi}$ can be prepared on a selected subset of vertices, denoted $I$.

Let $G$ be a graph with vertex set $V$ and edge set $\tilde{E}$. Let $I \subset V$ be a set of distinguished vertices such that $|I| = n$ and $|V \setminus I| = N$. Let $\ket{\psi}$ be a state in the Hilbert space associated to the input vertices, $\H_{I}$. Let $E$ denote the set of edges that are not entirely contained in $I$. The graph state with input is then
\begin{align}
\ket{G_{\psi}} := \prod_{\{v,v'\} \in E} \text{CZ}_{v,v'}\ket{\psi}_{I}\ket{+}^{\otimes N}.
\end{align}
For any input state $\ket{\psi}$, the graph state with input is invariant under the application of any operation in the set $\{K_{v} : v \in V \setminus I\}$ where
\begin{align}
K_{v} := X_{v} \otimes Z_{N_{v}^{G}}
\end{align}
with $N_{v}^{G}$ denoting the set of neighbours of vertex $v$ in $G$ and $Z_{N_{v}^{G}} := \bigotimes_{v' \in N_{v}^{G}} Z_{v'}$. The $K_{v}$ are all mutually independent and the set $\{K_{v} : v \in V \setminus I\}$ generates a $2^{n}$-dimensional subspace of $\mathcal{H}_{v}$, the graph codespace corresponding to $G$ (see e.g., \cite{schlingemann2002logical,Schlingemann_01} for further details on graph codes \footnote{Just as \cite{fellner_22} modified the original LHZ architecture \cite{lechner2015quantum} to contain data qubits, we are considering graph codes where the input qubits remain unmeasured, a modification to be removed in future work.}).

In the measurement-based regime, computation is driven by single-qubit projective measurements restricted to the $XY$-, $XZ$- and $YZ$-planes of the Bloch sphere. A given computation is defined by one specific outcome for each measurement, and the restriction to the given planes allows for the correction of undesired outcomes via an effective application of an appropriate stabiliser element (or products thereof). However, even with these restrictions not every sequence of measurements for a given graph states is possible. The combination of graph state and measurements that do allow for computation are well characterised by a property called gflow which is known to be a necessary and sufficient condition for deterministic MBQC \cite{browne2007generalized} (see \cite{supp} for the definition).


\section{Results}

It is known that every stabiliser code is equivalent to a graph code \cite{schlingemann2001stabilizer} (see also \cite{Nest_04}). The following is an instance of this result using the specific properties exhibited by parity codes.

\begin{Proposition} \label{prop:parity_is_bipartite} Every parity code is local Clifford equivalent to a bipartite graph code, where all data qubits are contained in one partition.
\end{Proposition}
\begin{proof} A parity code is defined by a stabiliser generated by the operators $K'_{(ij...k)} = Z_{(ij...k)} \otimes Z_{i} \otimes Z_{j} \otimes ... \otimes Z_{k}$ for each parity qubit $(ij...k)$. The set of qubits upon which these operators act includes only a single parity qubit. Via conjugation by Hadamards on each parity qubit, we obtain the local Clifford equivalent stabiliser generated by $K_{(ij...)} = X_{(ij...k)} \otimes Z_{i} \otimes Z_{j} \otimes ... \otimes Z_{k}$. Since the $K_{(ij...k)}$ are of the form $X_{(ij...k)} \otimes Z_{N_{(ij...k)}^{G}}$ which generate a graph code for a graph $G$ with edges between parity and data qubits. That is, each neighbourhood $N_{(ij...k)}^{G}$ contains only vertices corresponding to data qubits, which enforces a parity qubit-data qubit bipartition.
\end{proof}

An example of this correspondence is shown in \Cref{fig:LHZ_example} and \Cref{fig:graph_eg}. Note that there exist graph codes that are not local Clifford equivalent to bipartite graph codes, and hence are inequivalent to any parity code.

An immediate consequence of \Cref{prop:parity_is_bipartite} is that, for any $\ket{\psi}$, we have
\begin{align}
\ket{LHZ_{\psi}} = \bigotimes_{v \in V \setminus I}H_{v} \ket{G_{\psi}} \label{eq:PQC_MBQC_equal}
\end{align}
where $G$ denotes the bipartite graph corresponding to the parity code, $I$ denotes the set of vertices corresponding to data qubits and $V$ is the set of all qubits. We will use $V \setminus I$ to denote the set of parity qubits forthwith.

Consider the parity computation $U$ described in \Cref{eq:parity_comp}. For simplicity, let us first consider only the initial layer $l=1$ and drop the parameters $\boldsymbol{\alpha}$, $\boldsymbol{\phi}$ and $\boldsymbol{\theta}$ from the notation since the following holds for all parameter values. Using \Cref{eq:PQC_MBQC_equal}, we get that 
\begin{align}
U_{\text{data}}^{(1)}D_{\text{dec}}^{(1)}U_{\text{enc}} \ket{0}^{\otimes N}\ket{\psi} = U_{\text{data}}^{(1)}D_{\text{dec}}^{(1)} H_{V \setminus I} \ket{G_{\psi}}
\end{align}
where $H_{V \setminus I}$ is shorthand for $\otimes_{v \in V\setminus I} H_{v}$. Both $D_{\text{dec}}^{(1)}$ and $H_{V \setminus I}$ act on the same qubits and can be simplified as:
\begin{align}
D_{\text{dec}}^{(1)}H_{V\setminus I} = \bigotimes_{v \in V \setminus I} \bra{0_{v}}R_{X_{v}}(\theta_{v}^{(1)}).
\end{align}
The operator $\bra{0_{v}}R_{X_{v}}(\theta_{v}^{(1)})$ is a measurement in the YZ-plane of the Bloch sphere, and hence $D_{\text{dec}}^{(1)}H_{V \setminus I}\ket{G_{\psi}}$ is precisely a measurement-based computation where all measurements are restricted to that plane (the issue of measurement corrections is covered below). Denote the output of the first layer as $\ket{\psi^{(1)}}$, which is the resultant state of applying $U_{\text{data}}^{(1)}$ to the output state of the MBQC. The remaining computation is then given by
\begin{align}
\prod_{l =2}^{L} \left(U_{\text{data}}^{(l)} D_{\text{dec}}^{(l)}  U_{\text{enc}} \ket{0}^{\otimes N}\right) \ket{\psi^{(1)}}
\end{align}
for which the above process can be repeated. We have thus shown the following:

\begin{Theorem} \label{thm:parity_is_MBQC} Universal parity quantum computing is repeated measurement-based quantum computation using YZ-plane measurements, interleaved with local rotations.
\end{Theorem}

It should be noted that typically, MBQC is done on a fully pre-prepared graph state where input $I$ and output $O$ are distinct. However, proposals for repeated MBQC which de- and re-encode graph codes have been considered previously \cite{zwerger2014hybrid}; see also \cite{nautrup2023measurement} for a recent perspective.

In light of the above, it is prudent to demarcate the parity computing regime with respect to the MBQC regime. One could reasonably ask if there exist YZ-only measurement-based computations on graphs that are not bipartite. However, the following theorem demonstrates that this in fact not the case. The theorem also takes care of any issues regarding correction of measurements (see \cite{supp} for more details).

As mentioned above, MBQC on a graph state $G$ typically includes specifying an input and output set of vertices, denoted by $I$ and $O$ respectively. For a deterministic MBQC to produce a unitary transformation (as opposed to an isometry), we require $|I| = |O|$. 

\begin{Theorem} \label{thm:MBQC_YZ_I_O} MBQC on a (simple, connected) graph $G$ with $|I| = |O|$ and using only YZ-plane measurements is deterministic if and only if $G$ is bipartite with $I$ forming one partition.
\end{Theorem} 

The proof makes use of technical lemmas related to gflow which are proved along with the theorem in the Supplemental Material \cite{supp}.


\section{Discussion} 

This work has demonstrated that (i) parity codes are local Clifford equivalent to bipartite graph codes, (ii) as a consequence, parity quantum computing can be understood as repeated MBQC where all measurements are made in the YZ-plane, supplemented by local rotations, and (iii) MBQC with equivalent input and output qubits and using only YZ-plane measurements must use a bipartite graph state. 

Interestingly, these results demonstrate that the universal parity computing regime has effectively singled out YZ-plane unitary MBQC exactly. To the best of our knowledge, the restriction to having equivalent input and output and only YZ-measurements has not been considered before in the MBQC literature. On the other hand, this \textit{is} a restriction of the full MBQC framework, which means that there is ample scope for future investigation into what other aspects of MBQC could be brought to bear on the parity computing regime, and vice versa.

As this work connects two previously distinct bodies of research, there are a number of consequences worth mentioning here. Firstly, our results provide insight into recent research in the parity framework. In \cite{messinger2023constant}, it was noted that the parity measurement-based encoding and decoding procedures can be implemented in constant depth regardless of architecture size. Since the decoding procedure corresponds to measuring vertices in one partition of a bipartite graph, it is clear that all measurements can be done simultaneously and corrected for in the other partition. The encoding procedure can be understood as measuring ancilla vertices of a larger graph state in the $X$-basis, which in particular produces the required bipartite graph (see e.g., \cite{Hein_04,hein2006ent} for a characterisation of graph state deformations under Pauli measurements). It is known that all Pauli-measurements can be performed simultaneously in MBQC \cite{raussendorf2003measurement}.

Secondly, there are a number of potential avenues for future research enabled by the results presented here. In the MBQC literature, there exist multi-particle entanglement purification protocols for bipartite graph states which exhibit favourable error thresholds for realistic scenarios \cite{Dur_03,aschauer_05,zwerger_13}. Having demonstrated the prevalence of bipartite graph states in the parity framework, similar techniques may be of benefit for error mitigation in near-term implementations of parity quantum computations. Furthermore, proposals for fault-tolerant MBQC \cite{Bolt_16,nickerson2018measurement,Raussendorf_FT_07,RAUSSENDORF20062242} and universal blind quantum computation \cite{broadbent2009universal} could provide the foundation for fault-tolerant and cryptographic implementations of the parity framework. Conversely, developments related to quantum optimisation within the parity framework \cite{Lanthaler_22,Dlaska_22,Ender_22} could inspire similar developments in MBQC where application to optimisation problems remains relatively unexplored. A number of these avenues are already being pursued in separate work.

\begin{acknowledgments} We would like to thank B. Klaver and A. Messinger for useful discussions. This research was funded in whole or in part by the Austrian Science Fund (FWF) through the DK-ALM W$1259$-N$27$ and the SFB BeyondC F7102. For open access purposes, the authors have applied a CC BY public copyright license to any author-accepted manuscript version arising from this submission. This work was also co-funded by the European Union (ERC, QuantAI, Project No. $101055129$). Views and opinions expressed are however those of the author(s) only and do not necessarily reflect those of the European Union or the European Research Council. Neither the European Union nor the granting authority can be held responsible for them.
\end{acknowledgments}

\bibliography{MBQC_PQC}

\begin{thebibliography}{37}%
\makeatletter
\providecommand \@ifxundefined [1]{%
 \@ifx{#1\undefined}
}%
\providecommand \@ifnum [1]{%
 \ifnum #1\expandafter \@firstoftwo
 \else \expandafter \@secondoftwo
 \fi
}%
\providecommand \@ifx [1]{%
 \ifx #1\expandafter \@firstoftwo
 \else \expandafter \@secondoftwo
 \fi
}%
\providecommand \natexlab [1]{#1}%
\providecommand \enquote  [1]{``#1''}%
\providecommand \bibnamefont  [1]{#1}%
\providecommand \bibfnamefont [1]{#1}%
\providecommand \citenamefont [1]{#1}%
\providecommand \href@noop [0]{\@secondoftwo}%
\providecommand \href [0]{\begingroup \@sanitize@url \@href}%
\providecommand \@href[1]{\@@startlink{#1}\@@href}%
\providecommand \@@href[1]{\endgroup#1\@@endlink}%
\providecommand \@sanitize@url [0]{\catcode `\\12\catcode `\$12\catcode `\&12\catcode `\#12\catcode `\^12\catcode `\_12\catcode `\%12\relax}%
\providecommand \@@startlink[1]{}%
\providecommand \@@endlink[0]{}%
\providecommand \url  [0]{\begingroup\@sanitize@url \@url }%
\providecommand \@url [1]{\endgroup\@href {#1}{\urlprefix }}%
\providecommand \urlprefix  [0]{URL }%
\providecommand \Eprint [0]{\href }%
\providecommand \doibase [0]{https://doi.org/}%
\providecommand \selectlanguage [0]{\@gobble}%
\providecommand \bibinfo  [0]{\@secondoftwo}%
\providecommand \bibfield  [0]{\@secondoftwo}%
\providecommand \translation [1]{[#1]}%
\providecommand \BibitemOpen [0]{}%
\providecommand \bibitemStop [0]{}%
\providecommand \bibitemNoStop [0]{.\EOS\space}%
\providecommand \EOS [0]{\spacefactor3000\relax}%
\providecommand \BibitemShut  [1]{\csname bibitem#1\endcsname}%
\let\auto@bib@innerbib\@empty
\bibitem [{\citenamefont {Preskill}(2018)}]{Preskill2018quantumcomputingin}%
  \BibitemOpen
  \bibfield  {author} {\bibinfo {author} {\bibfnamefont {J.}~\bibnamefont {Preskill}},\ }\bibfield  {title} {\bibinfo {title} {Quantum {C}omputing in the {NISQ} era and beyond},\ }\href {https://doi.org/10.22331/q-2018-08-06-79} {\bibfield  {journal} {\bibinfo  {journal} {{Quantum}}\ }\textbf {\bibinfo {volume} {2}},\ \bibinfo {pages} {79} (\bibinfo {year} {2018})}\BibitemShut {NoStop}%
\bibitem [{\citenamefont {Lechner}\ \emph {et~al.}(2015)\citenamefont {Lechner}, \citenamefont {Hauke},\ and\ \citenamefont {Zoller}}]{lechner2015quantum}%
  \BibitemOpen
  \bibfield  {author} {\bibinfo {author} {\bibfnamefont {W.}~\bibnamefont {Lechner}}, \bibinfo {author} {\bibfnamefont {P.}~\bibnamefont {Hauke}},\ and\ \bibinfo {author} {\bibfnamefont {P.}~\bibnamefont {Zoller}},\ }\bibfield  {title} {\bibinfo {title} {A quantum annealing architecture with all-to-all connectivity from local interactions},\ }\href@noop {} {\bibfield  {journal} {\bibinfo  {journal} {Science advances}\ }\textbf {\bibinfo {volume} {1}},\ \bibinfo {pages} {e1500838} (\bibinfo {year} {2015})}\BibitemShut {NoStop}%
\bibitem [{\citenamefont {Fellner}\ \emph {et~al.}(2022{\natexlab{a}})\citenamefont {Fellner}, \citenamefont {Messinger}, \citenamefont {Ender},\ and\ \citenamefont {Lechner}}]{fellner_22}%
  \BibitemOpen
  \bibfield  {author} {\bibinfo {author} {\bibfnamefont {M.}~\bibnamefont {Fellner}}, \bibinfo {author} {\bibfnamefont {A.}~\bibnamefont {Messinger}}, \bibinfo {author} {\bibfnamefont {K.}~\bibnamefont {Ender}},\ and\ \bibinfo {author} {\bibfnamefont {W.}~\bibnamefont {Lechner}},\ }\bibfield  {title} {\bibinfo {title} {Universal parity quantum computing},\ }\href {https://doi.org/10.1103/PhysRevLett.129.180503} {\bibfield  {journal} {\bibinfo  {journal} {Phys. Rev. Lett.}\ }\textbf {\bibinfo {volume} {129}},\ \bibinfo {pages} {180503} (\bibinfo {year} {2022}{\natexlab{a}})}\BibitemShut {NoStop}%
\bibitem [{\citenamefont {Fellner}\ \emph {et~al.}(2022{\natexlab{b}})\citenamefont {Fellner}, \citenamefont {Messinger}, \citenamefont {Ender},\ and\ \citenamefont {Lechner}}]{Fellner_PRA_22}%
  \BibitemOpen
  \bibfield  {author} {\bibinfo {author} {\bibfnamefont {M.}~\bibnamefont {Fellner}}, \bibinfo {author} {\bibfnamefont {A.}~\bibnamefont {Messinger}}, \bibinfo {author} {\bibfnamefont {K.}~\bibnamefont {Ender}},\ and\ \bibinfo {author} {\bibfnamefont {W.}~\bibnamefont {Lechner}},\ }\bibfield  {title} {\bibinfo {title} {Applications of universal parity quantum computation},\ }\href {https://doi.org/10.1103/PhysRevA.106.042442} {\bibfield  {journal} {\bibinfo  {journal} {Phys. Rev. A}\ }\textbf {\bibinfo {volume} {106}},\ \bibinfo {pages} {042442} (\bibinfo {year} {2022}{\natexlab{b}})}\BibitemShut {NoStop}%
\bibitem [{\citenamefont {Fellner}\ \emph {et~al.}(2023)\citenamefont {Fellner}, \citenamefont {Ender}, \citenamefont {ter Hoeven},\ and\ \citenamefont {Lechner}}]{Fellner2023parityquantum}%
  \BibitemOpen
  \bibfield  {author} {\bibinfo {author} {\bibfnamefont {M.}~\bibnamefont {Fellner}}, \bibinfo {author} {\bibfnamefont {K.}~\bibnamefont {Ender}}, \bibinfo {author} {\bibfnamefont {R.}~\bibnamefont {ter Hoeven}},\ and\ \bibinfo {author} {\bibfnamefont {W.}~\bibnamefont {Lechner}},\ }\bibfield  {title} {\bibinfo {title} {Parity {Q}uantum {O}ptimization: {B}enchmarks},\ }\href {https://doi.org/10.22331/q-2023-03-17-952} {\bibfield  {journal} {\bibinfo  {journal} {{Quantum}}\ }\textbf {\bibinfo {volume} {7}},\ \bibinfo {pages} {952} (\bibinfo {year} {2023})}\BibitemShut {NoStop}%
\bibitem [{\citenamefont {Drieb-Sch{\"{o}}n}\ \emph {et~al.}(2023)\citenamefont {Drieb-Sch{\"{o}}n}, \citenamefont {Ender}, \citenamefont {Javanmard},\ and\ \citenamefont {Lechner}}]{DriebSchon2023parityquantum}%
  \BibitemOpen
  \bibfield  {author} {\bibinfo {author} {\bibfnamefont {M.}~\bibnamefont {Drieb-Sch{\"{o}}n}}, \bibinfo {author} {\bibfnamefont {K.}~\bibnamefont {Ender}}, \bibinfo {author} {\bibfnamefont {Y.}~\bibnamefont {Javanmard}},\ and\ \bibinfo {author} {\bibfnamefont {W.}~\bibnamefont {Lechner}},\ }\bibfield  {title} {\bibinfo {title} {Parity {Q}uantum {O}ptimization: {E}ncoding {C}onstraints},\ }\href {https://doi.org/10.22331/q-2023-03-17-951} {\bibfield  {journal} {\bibinfo  {journal} {{Quantum}}\ }\textbf {\bibinfo {volume} {7}},\ \bibinfo {pages} {951} (\bibinfo {year} {2023})}\BibitemShut {NoStop}%
\bibitem [{\citenamefont {Ender}\ \emph {et~al.}(2023)\citenamefont {Ender}, \citenamefont {ter Hoeven}, \citenamefont {Niehoff}, \citenamefont {Drieb-Sch{\"{o}}n},\ and\ \citenamefont {Lechner}}]{Ender2023parityquantum}%
  \BibitemOpen
  \bibfield  {author} {\bibinfo {author} {\bibfnamefont {K.}~\bibnamefont {Ender}}, \bibinfo {author} {\bibfnamefont {R.}~\bibnamefont {ter Hoeven}}, \bibinfo {author} {\bibfnamefont {B.~E.}\ \bibnamefont {Niehoff}}, \bibinfo {author} {\bibfnamefont {M.}~\bibnamefont {Drieb-Sch{\"{o}}n}},\ and\ \bibinfo {author} {\bibfnamefont {W.}~\bibnamefont {Lechner}},\ }\bibfield  {title} {\bibinfo {title} {Parity {Q}uantum {O}ptimization: {C}ompiler},\ }\href {https://doi.org/10.22331/q-2023-03-17-950} {\bibfield  {journal} {\bibinfo  {journal} {{Quantum}}\ }\textbf {\bibinfo {volume} {7}},\ \bibinfo {pages} {950} (\bibinfo {year} {2023})}\BibitemShut {NoStop}%
\bibitem [{\citenamefont {Gottesman}(1997)}]{gottesman1997stabilizer}%
  \BibitemOpen
  \bibfield  {author} {\bibinfo {author} {\bibfnamefont {D.}~\bibnamefont {Gottesman}},\ }\href {https://doi.org/https://doi.org/10.48550/arXiv.quant-ph/9705052} {\emph {\bibinfo {title} {Stabilizer codes and quantum error correction}}}\ (\bibinfo  {publisher} {California Institute of Technology},\ \bibinfo {year} {1997})\BibitemShut {NoStop}%
\bibitem [{\citenamefont {Nielsen}\ and\ \citenamefont {Chuang}(2010)}]{nielsen_chuang_2010}%
  \BibitemOpen
  \bibfield  {author} {\bibinfo {author} {\bibfnamefont {M.~A.}\ \bibnamefont {Nielsen}}\ and\ \bibinfo {author} {\bibfnamefont {I.~L.}\ \bibnamefont {Chuang}},\ }\href {https://doi.org/https://doi.org/10.1017/CBO9780511976667} {\emph {\bibinfo {title} {Quantum Computation and Quantum Information: 10th Anniversary Edition}}}\ (\bibinfo  {publisher} {Cambridge University Press},\ \bibinfo {year} {2010})\BibitemShut {NoStop}%
\bibitem [{\citenamefont {Hein}\ \emph {et~al.}(2004)\citenamefont {Hein}, \citenamefont {Eisert},\ and\ \citenamefont {Briegel}}]{Hein_04}%
  \BibitemOpen
  \bibfield  {author} {\bibinfo {author} {\bibfnamefont {M.}~\bibnamefont {Hein}}, \bibinfo {author} {\bibfnamefont {J.}~\bibnamefont {Eisert}},\ and\ \bibinfo {author} {\bibfnamefont {H.~J.}\ \bibnamefont {Briegel}},\ }\bibfield  {title} {\bibinfo {title} {Multiparty entanglement in graph states},\ }\href {https://doi.org/10.1103/PhysRevA.69.062311} {\bibfield  {journal} {\bibinfo  {journal} {Phys. Rev. A}\ }\textbf {\bibinfo {volume} {69}},\ \bibinfo {pages} {062311} (\bibinfo {year} {2004})}\BibitemShut {NoStop}%
\bibitem [{\citenamefont {Hein}\ \emph {et~al.}(2006)\citenamefont {Hein}, \citenamefont {D{\"u}r}, \citenamefont {Eisert}, \citenamefont {Raussendorf}, \citenamefont {Van~den Nest},\ and\ \citenamefont {Briegel}}]{hein2006ent}%
  \BibitemOpen
  \bibfield  {author} {\bibinfo {author} {\bibfnamefont {M.}~\bibnamefont {Hein}}, \bibinfo {author} {\bibfnamefont {W.}~\bibnamefont {D{\"u}r}}, \bibinfo {author} {\bibfnamefont {J.}~\bibnamefont {Eisert}}, \bibinfo {author} {\bibfnamefont {R.}~\bibnamefont {Raussendorf}}, \bibinfo {author} {\bibfnamefont {M.}~\bibnamefont {Van~den Nest}},\ and\ \bibinfo {author} {\bibfnamefont {H.~J.}\ \bibnamefont {Briegel}},\ }\bibfield  {title} {\bibinfo {title} {Entanglement in graph states and its applications},\ }in\ \href {https://doi.org/https://doi.org/10.3254/978-1-61499-018-5-115} {\emph {\bibinfo {booktitle} {Volume 162: Quantum Computers, Algorithms and Chaos}}},\ \bibinfo {series and number} {Proceedings of the International School of Physics "Enrico Fermi"}\ (\bibinfo  {publisher} {IOS Press Ebooks},\ \bibinfo {year} {2006})\ pp.\ \bibinfo {pages} {115--218}\BibitemShut {NoStop}%
\bibitem [{\citenamefont {Schlingemann}(2002)}]{schlingemann2002logical}%
  \BibitemOpen
  \bibfield  {author} {\bibinfo {author} {\bibfnamefont {D.}~\bibnamefont {Schlingemann}},\ }\bibfield  {title} {\bibinfo {title} {Logical network implementation for cluster states and graph codes},\ }\href@noop {} {\bibfield  {journal} {\bibinfo  {journal} {arXiv preprint quant-ph/0202007}\ } (\bibinfo {year} {2002})}\BibitemShut {NoStop}%
\bibitem [{\citenamefont {Schlingemann}(2001)}]{schlingemann2001stabilizer}%
  \BibitemOpen
  \bibfield  {author} {\bibinfo {author} {\bibfnamefont {D.}~\bibnamefont {Schlingemann}},\ }\bibfield  {title} {\bibinfo {title} {Stabilizer codes can be realized as graph codes},\ }\href@noop {} {\bibfield  {journal} {\bibinfo  {journal} {arXiv preprint quant-ph/0111080}\ } (\bibinfo {year} {2001})}\BibitemShut {NoStop}%
\bibitem [{\citenamefont {Schlingemann}\ and\ \citenamefont {Werner}(2001)}]{Schlingemann_01}%
  \BibitemOpen
  \bibfield  {author} {\bibinfo {author} {\bibfnamefont {D.}~\bibnamefont {Schlingemann}}\ and\ \bibinfo {author} {\bibfnamefont {R.~F.}\ \bibnamefont {Werner}},\ }\bibfield  {title} {\bibinfo {title} {Quantum error-correcting codes associated with graphs},\ }\href {https://doi.org/10.1103/PhysRevA.65.012308} {\bibfield  {journal} {\bibinfo  {journal} {Phys. Rev. A}\ }\textbf {\bibinfo {volume} {65}},\ \bibinfo {pages} {012308} (\bibinfo {year} {2001})}\BibitemShut {NoStop}%
\bibitem [{\citenamefont {Raussendorf}\ and\ \citenamefont {Briegel}(2001{\natexlab{a}})}]{raussendorf2001one}%
  \BibitemOpen
  \bibfield  {author} {\bibinfo {author} {\bibfnamefont {R.}~\bibnamefont {Raussendorf}}\ and\ \bibinfo {author} {\bibfnamefont {H.~J.}\ \bibnamefont {Briegel}},\ }\bibfield  {title} {\bibinfo {title} {A one-way quantum computer},\ }\href {https://doi.org/https://doi.org/10.1103/PhysRevLett.86.5188} {\bibfield  {journal} {\bibinfo  {journal} {Phys. Rev. Lett.}\ }\textbf {\bibinfo {volume} {86}},\ \bibinfo {pages} {5188} (\bibinfo {year} {2001}{\natexlab{a}})}\BibitemShut {NoStop}%
\bibitem [{\citenamefont {Briegel}\ \emph {et~al.}(2009)\citenamefont {Briegel}, \citenamefont {Browne}, \citenamefont {D{\"u}r}, \citenamefont {Raussendorf},\ and\ \citenamefont {Van~den Nest}}]{briegel2009measurement}%
  \BibitemOpen
  \bibfield  {author} {\bibinfo {author} {\bibfnamefont {H.~J.}\ \bibnamefont {Briegel}}, \bibinfo {author} {\bibfnamefont {D.~E.}\ \bibnamefont {Browne}}, \bibinfo {author} {\bibfnamefont {W.}~\bibnamefont {D{\"u}r}}, \bibinfo {author} {\bibfnamefont {R.}~\bibnamefont {Raussendorf}},\ and\ \bibinfo {author} {\bibfnamefont {M.}~\bibnamefont {Van~den Nest}},\ }\bibfield  {title} {\bibinfo {title} {Measurement-based quantum computation},\ }\href {https://doi.org/https://doi.org/10.1038/nphys1157} {\bibfield  {journal} {\bibinfo  {journal} {Nature Physics}\ }\textbf {\bibinfo {volume} {5}},\ \bibinfo {pages} {19} (\bibinfo {year} {2009})}\BibitemShut {NoStop}%
\bibitem [{\citenamefont {Raussendorf}\ and\ \citenamefont {Briegel}(2001{\natexlab{b}})}]{raussendorf2001computational}%
  \BibitemOpen
  \bibfield  {author} {\bibinfo {author} {\bibfnamefont {R.}~\bibnamefont {Raussendorf}}\ and\ \bibinfo {author} {\bibfnamefont {H.}~\bibnamefont {Briegel}},\ }\bibfield  {title} {\bibinfo {title} {Computational model underlying the one-way quantum computer},\ }\bibfield  {journal} {\bibinfo  {journal} {arXiv preprint quant-ph/0108067}\ }\href {https://doi.org/https://doi.org/10.48550/arXiv.quant-ph/0108067} {https://doi.org/10.48550/arXiv.quant-ph/0108067} (\bibinfo {year} {2001}{\natexlab{b}})\BibitemShut {NoStop}%
\bibitem [{\citenamefont {Raussendorf}\ \emph {et~al.}(2003)\citenamefont {Raussendorf}, \citenamefont {Browne},\ and\ \citenamefont {Briegel}}]{raussendorf2003measurement}%
  \BibitemOpen
  \bibfield  {author} {\bibinfo {author} {\bibfnamefont {R.}~\bibnamefont {Raussendorf}}, \bibinfo {author} {\bibfnamefont {D.~E.}\ \bibnamefont {Browne}},\ and\ \bibinfo {author} {\bibfnamefont {H.~J.}\ \bibnamefont {Briegel}},\ }\bibfield  {title} {\bibinfo {title} {Measurement-based quantum computation on cluster states},\ }\href {https://doi.org/https://doi.org/10.1103/PhysRevA.68.022312} {\bibfield  {journal} {\bibinfo  {journal} {Phys. Rev. A}\ }\textbf {\bibinfo {volume} {68}},\ \bibinfo {pages} {022312} (\bibinfo {year} {2003})}\BibitemShut {NoStop}%
\bibitem [{\citenamefont {Walther}\ \emph {et~al.}(2005)\citenamefont {Walther}, \citenamefont {Resch}, \citenamefont {Rudolph}, \citenamefont {Schenck}, \citenamefont {Weinfurter}, \citenamefont {Vedral}, \citenamefont {Aspelmeyer},\ and\ \citenamefont {Zeilinger}}]{walther2005experimental}%
  \BibitemOpen
  \bibfield  {author} {\bibinfo {author} {\bibfnamefont {P.}~\bibnamefont {Walther}}, \bibinfo {author} {\bibfnamefont {K.~J.}\ \bibnamefont {Resch}}, \bibinfo {author} {\bibfnamefont {T.}~\bibnamefont {Rudolph}}, \bibinfo {author} {\bibfnamefont {E.}~\bibnamefont {Schenck}}, \bibinfo {author} {\bibfnamefont {H.}~\bibnamefont {Weinfurter}}, \bibinfo {author} {\bibfnamefont {V.}~\bibnamefont {Vedral}}, \bibinfo {author} {\bibfnamefont {M.}~\bibnamefont {Aspelmeyer}},\ and\ \bibinfo {author} {\bibfnamefont {A.}~\bibnamefont {Zeilinger}},\ }\bibfield  {title} {\bibinfo {title} {Experimental one-way quantum computing},\ }\href@noop {} {\bibfield  {journal} {\bibinfo  {journal} {Nature}\ }\textbf {\bibinfo {volume} {434}},\ \bibinfo {pages} {169} (\bibinfo {year} {2005})}\BibitemShut {NoStop}%
\bibitem [{\citenamefont {{Poulsen Nautrup}}\ and\ \citenamefont {Briegel}(2023)}]{nautrup2023measurement}%
  \BibitemOpen
  \bibfield  {author} {\bibinfo {author} {\bibfnamefont {H.}~\bibnamefont {{Poulsen Nautrup}}}\ and\ \bibinfo {author} {\bibfnamefont {H.~J.}\ \bibnamefont {Briegel}},\ }\bibfield  {title} {\bibinfo {title} {Measurement-based quantum computation from clifford quantum cellular automata},\ }\href@noop {} {\bibfield  {journal} {\bibinfo  {journal} {arXiv preprint arXiv:2312.13185}\ } (\bibinfo {year} {2023})}\BibitemShut {NoStop}%
\bibitem [{\citenamefont {Messinger}\ \emph {et~al.}(2023)\citenamefont {Messinger}, \citenamefont {Fellner},\ and\ \citenamefont {Lechner}}]{messinger2023constant}%
  \BibitemOpen
  \bibfield  {author} {\bibinfo {author} {\bibfnamefont {A.}~\bibnamefont {Messinger}}, \bibinfo {author} {\bibfnamefont {M.}~\bibnamefont {Fellner}},\ and\ \bibinfo {author} {\bibfnamefont {W.}~\bibnamefont {Lechner}},\ }\bibfield  {title} {\bibinfo {title} {Constant depth code deformations in the parity architecture},\ }in\ \href {https://doi.org/10.1109/QCE57702.2023.00022} {\emph {\bibinfo {booktitle} {2023 IEEE International Conference on Quantum Computing and Engineering (QCE)}}},\ Vol.~\bibinfo {volume} {01}\ (\bibinfo {year} {2023})\ pp.\ \bibinfo {pages} {120--130}\BibitemShut {NoStop}%
\bibitem [{\citenamefont {Browne}\ \emph {et~al.}(2007)\citenamefont {Browne}, \citenamefont {Kashefi}, \citenamefont {Mhalla},\ and\ \citenamefont {Perdrix}}]{browne2007generalized}%
  \BibitemOpen
  \bibfield  {author} {\bibinfo {author} {\bibfnamefont {D.~E.}\ \bibnamefont {Browne}}, \bibinfo {author} {\bibfnamefont {E.}~\bibnamefont {Kashefi}}, \bibinfo {author} {\bibfnamefont {M.}~\bibnamefont {Mhalla}},\ and\ \bibinfo {author} {\bibfnamefont {S.}~\bibnamefont {Perdrix}},\ }\bibfield  {title} {\bibinfo {title} {Generalized flow and determinism in measurement-based quantum computation},\ }\href@noop {} {\bibfield  {journal} {\bibinfo  {journal} {New Journal of Physics}\ }\textbf {\bibinfo {volume} {9}},\ \bibinfo {pages} {250} (\bibinfo {year} {2007})}\BibitemShut {NoStop}%
\bibitem [{Note1()}]{Note1}%
  \BibitemOpen
  \bibinfo {note} {Just as \cite {fellner_22} modified the original LHZ architecture \cite {lechner2015quantum} to contain data qubits, we are considering graph codes where the input qubits remain unmeasured, a modification to be removed in future work.}\BibitemShut {Stop}%
\bibitem [{sup()}]{supp}%
  \BibitemOpen
  \href@noop {} {}\bibinfo {note} {See Supplemental Material at URL-will-be-inserted-by-publisher for the definition of gflow and the proof of Theorem $2$.}\BibitemShut {Stop}%
\bibitem [{\citenamefont {Van~den Nest}\ \emph {et~al.}(2004)\citenamefont {Van~den Nest}, \citenamefont {Dehaene},\ and\ \citenamefont {De~Moor}}]{Nest_04}%
  \BibitemOpen
  \bibfield  {author} {\bibinfo {author} {\bibfnamefont {M.}~\bibnamefont {Van~den Nest}}, \bibinfo {author} {\bibfnamefont {J.}~\bibnamefont {Dehaene}},\ and\ \bibinfo {author} {\bibfnamefont {B.}~\bibnamefont {De~Moor}},\ }\bibfield  {title} {\bibinfo {title} {Graphical description of the action of local clifford transformations on graph states},\ }\href {https://doi.org/10.1103/PhysRevA.69.022316} {\bibfield  {journal} {\bibinfo  {journal} {Phys. Rev. A}\ }\textbf {\bibinfo {volume} {69}},\ \bibinfo {pages} {022316} (\bibinfo {year} {2004})}\BibitemShut {NoStop}%
\bibitem [{\citenamefont {Zwerger}\ \emph {et~al.}(2014)\citenamefont {Zwerger}, \citenamefont {Briegel},\ and\ \citenamefont {D{\"u}r}}]{zwerger2014hybrid}%
  \BibitemOpen
  \bibfield  {author} {\bibinfo {author} {\bibfnamefont {M.}~\bibnamefont {Zwerger}}, \bibinfo {author} {\bibfnamefont {H.}~\bibnamefont {Briegel}},\ and\ \bibinfo {author} {\bibfnamefont {W.}~\bibnamefont {D{\"u}r}},\ }\bibfield  {title} {\bibinfo {title} {Hybrid architecture for encoded measurement-based quantum computation},\ }\href@noop {} {\bibfield  {journal} {\bibinfo  {journal} {Scientific reports}\ }\textbf {\bibinfo {volume} {4}},\ \bibinfo {pages} {5364} (\bibinfo {year} {2014})}\BibitemShut {NoStop}%
\bibitem [{\citenamefont {D\"ur}\ \emph {et~al.}(2003)\citenamefont {D\"ur}, \citenamefont {Aschauer},\ and\ \citenamefont {Briegel}}]{Dur_03}%
  \BibitemOpen
  \bibfield  {author} {\bibinfo {author} {\bibfnamefont {W.}~\bibnamefont {D\"ur}}, \bibinfo {author} {\bibfnamefont {H.}~\bibnamefont {Aschauer}},\ and\ \bibinfo {author} {\bibfnamefont {H.-J.}\ \bibnamefont {Briegel}},\ }\bibfield  {title} {\bibinfo {title} {Multiparticle entanglement purification for graph states},\ }\href {https://doi.org/10.1103/PhysRevLett.91.107903} {\bibfield  {journal} {\bibinfo  {journal} {Phys. Rev. Lett.}\ }\textbf {\bibinfo {volume} {91}},\ \bibinfo {pages} {107903} (\bibinfo {year} {2003})}\BibitemShut {NoStop}%
\bibitem [{\citenamefont {Aschauer}\ \emph {et~al.}(2005)\citenamefont {Aschauer}, \citenamefont {D\"ur},\ and\ \citenamefont {Briegel}}]{aschauer_05}%
  \BibitemOpen
  \bibfield  {author} {\bibinfo {author} {\bibfnamefont {H.}~\bibnamefont {Aschauer}}, \bibinfo {author} {\bibfnamefont {W.}~\bibnamefont {D\"ur}},\ and\ \bibinfo {author} {\bibfnamefont {H.-J.}\ \bibnamefont {Briegel}},\ }\bibfield  {title} {\bibinfo {title} {Multiparticle entanglement purification for two-colorable graph states},\ }\href {https://doi.org/10.1103/PhysRevA.71.012319} {\bibfield  {journal} {\bibinfo  {journal} {Phys. Rev. A}\ }\textbf {\bibinfo {volume} {71}},\ \bibinfo {pages} {012319} (\bibinfo {year} {2005})}\BibitemShut {NoStop}%
\bibitem [{\citenamefont {Zwerger}\ \emph {et~al.}(2013)\citenamefont {Zwerger}, \citenamefont {Briegel},\ and\ \citenamefont {D\"ur}}]{zwerger_13}%
  \BibitemOpen
  \bibfield  {author} {\bibinfo {author} {\bibfnamefont {M.}~\bibnamefont {Zwerger}}, \bibinfo {author} {\bibfnamefont {H.~J.}\ \bibnamefont {Briegel}},\ and\ \bibinfo {author} {\bibfnamefont {W.}~\bibnamefont {D\"ur}},\ }\bibfield  {title} {\bibinfo {title} {Universal and optimal error thresholds for measurement-based entanglement purification},\ }\href {https://doi.org/10.1103/PhysRevLett.110.260503} {\bibfield  {journal} {\bibinfo  {journal} {Phys. Rev. Lett.}\ }\textbf {\bibinfo {volume} {110}},\ \bibinfo {pages} {260503} (\bibinfo {year} {2013})}\BibitemShut {NoStop}%
\bibitem [{\citenamefont {Bolt}\ \emph {et~al.}(2016)\citenamefont {Bolt}, \citenamefont {Duclos-Cianci}, \citenamefont {Poulin},\ and\ \citenamefont {Stace}}]{Bolt_16}%
  \BibitemOpen
  \bibfield  {author} {\bibinfo {author} {\bibfnamefont {A.}~\bibnamefont {Bolt}}, \bibinfo {author} {\bibfnamefont {G.}~\bibnamefont {Duclos-Cianci}}, \bibinfo {author} {\bibfnamefont {D.}~\bibnamefont {Poulin}},\ and\ \bibinfo {author} {\bibfnamefont {T.~M.}\ \bibnamefont {Stace}},\ }\bibfield  {title} {\bibinfo {title} {Foliated quantum error-correcting codes},\ }\href {https://doi.org/10.1103/PhysRevLett.117.070501} {\bibfield  {journal} {\bibinfo  {journal} {Phys. Rev. Lett.}\ }\textbf {\bibinfo {volume} {117}},\ \bibinfo {pages} {070501} (\bibinfo {year} {2016})}\BibitemShut {NoStop}%
\bibitem [{\citenamefont {Nickerson}\ and\ \citenamefont {Bomb{\'\i}n}(2018)}]{nickerson2018measurement}%
  \BibitemOpen
  \bibfield  {author} {\bibinfo {author} {\bibfnamefont {N.}~\bibnamefont {Nickerson}}\ and\ \bibinfo {author} {\bibfnamefont {H.}~\bibnamefont {Bomb{\'\i}n}},\ }\bibfield  {title} {\bibinfo {title} {Measurement based fault tolerance beyond foliation},\ }\href@noop {} {\bibfield  {journal} {\bibinfo  {journal} {arXiv preprint arXiv:1810.09621}\ } (\bibinfo {year} {2018})}\BibitemShut {NoStop}%
\bibitem [{\citenamefont {Raussendorf}\ and\ \citenamefont {Harrington}(2007)}]{Raussendorf_FT_07}%
  \BibitemOpen
  \bibfield  {author} {\bibinfo {author} {\bibfnamefont {R.}~\bibnamefont {Raussendorf}}\ and\ \bibinfo {author} {\bibfnamefont {J.}~\bibnamefont {Harrington}},\ }\bibfield  {title} {\bibinfo {title} {Fault-tolerant quantum computation with high threshold in two dimensions},\ }\href {https://doi.org/10.1103/PhysRevLett.98.190504} {\bibfield  {journal} {\bibinfo  {journal} {Phys. Rev. Lett.}\ }\textbf {\bibinfo {volume} {98}},\ \bibinfo {pages} {190504} (\bibinfo {year} {2007})}\BibitemShut {NoStop}%
\bibitem [{\citenamefont {Raussendorf}\ \emph {et~al.}(2006)\citenamefont {Raussendorf}, \citenamefont {Harrington},\ and\ \citenamefont {Goyal}}]{RAUSSENDORF20062242}%
  \BibitemOpen
  \bibfield  {author} {\bibinfo {author} {\bibfnamefont {R.}~\bibnamefont {Raussendorf}}, \bibinfo {author} {\bibfnamefont {J.}~\bibnamefont {Harrington}},\ and\ \bibinfo {author} {\bibfnamefont {K.}~\bibnamefont {Goyal}},\ }\bibfield  {title} {\bibinfo {title} {A fault-tolerant one-way quantum computer},\ }\href {https://doi.org/https://doi.org/10.1016/j.aop.2006.01.012} {\bibfield  {journal} {\bibinfo  {journal} {Annals of Physics}\ }\textbf {\bibinfo {volume} {321}},\ \bibinfo {pages} {2242} (\bibinfo {year} {2006})}\BibitemShut {NoStop}%
\bibitem [{\citenamefont {Broadbent}\ \emph {et~al.}(2009)\citenamefont {Broadbent}, \citenamefont {Fitzsimons},\ and\ \citenamefont {Kashefi}}]{broadbent2009universal}%
  \BibitemOpen
  \bibfield  {author} {\bibinfo {author} {\bibfnamefont {A.}~\bibnamefont {Broadbent}}, \bibinfo {author} {\bibfnamefont {J.}~\bibnamefont {Fitzsimons}},\ and\ \bibinfo {author} {\bibfnamefont {E.}~\bibnamefont {Kashefi}},\ }\bibfield  {title} {\bibinfo {title} {Universal blind quantum computation},\ }in\ \href@noop {} {\emph {\bibinfo {booktitle} {2009 50th annual IEEE symposium on foundations of computer science}}}\ (\bibinfo {organization} {IEEE},\ \bibinfo {year} {2009})\ pp.\ \bibinfo {pages} {517--526}\BibitemShut {NoStop}%
\bibitem [{\citenamefont {Lanthaler}\ \emph {et~al.}(2023)\citenamefont {Lanthaler}, \citenamefont {Dlaska}, \citenamefont {Ender},\ and\ \citenamefont {Lechner}}]{Lanthaler_22}%
  \BibitemOpen
  \bibfield  {author} {\bibinfo {author} {\bibfnamefont {M.}~\bibnamefont {Lanthaler}}, \bibinfo {author} {\bibfnamefont {C.}~\bibnamefont {Dlaska}}, \bibinfo {author} {\bibfnamefont {K.}~\bibnamefont {Ender}},\ and\ \bibinfo {author} {\bibfnamefont {W.}~\bibnamefont {Lechner}},\ }\bibfield  {title} {\bibinfo {title} {Rydberg-blockade-based parity quantum optimization},\ }\href {https://doi.org/10.1103/PhysRevLett.130.220601} {\bibfield  {journal} {\bibinfo  {journal} {Phys. Rev. Lett.}\ }\textbf {\bibinfo {volume} {130}},\ \bibinfo {pages} {220601} (\bibinfo {year} {2023})}\BibitemShut {NoStop}%
\bibitem [{\citenamefont {Dlaska}\ \emph {et~al.}(2022)\citenamefont {Dlaska}, \citenamefont {Ender}, \citenamefont {Mbeng}, \citenamefont {Kruckenhauser}, \citenamefont {Lechner},\ and\ \citenamefont {van Bijnen}}]{Dlaska_22}%
  \BibitemOpen
  \bibfield  {author} {\bibinfo {author} {\bibfnamefont {C.}~\bibnamefont {Dlaska}}, \bibinfo {author} {\bibfnamefont {K.}~\bibnamefont {Ender}}, \bibinfo {author} {\bibfnamefont {G.~B.}\ \bibnamefont {Mbeng}}, \bibinfo {author} {\bibfnamefont {A.}~\bibnamefont {Kruckenhauser}}, \bibinfo {author} {\bibfnamefont {W.}~\bibnamefont {Lechner}},\ and\ \bibinfo {author} {\bibfnamefont {R.}~\bibnamefont {van Bijnen}},\ }\bibfield  {title} {\bibinfo {title} {Quantum optimization via four-body rydberg gates},\ }\href {https://doi.org/10.1103/PhysRevLett.128.120503} {\bibfield  {journal} {\bibinfo  {journal} {Phys. Rev. Lett.}\ }\textbf {\bibinfo {volume} {128}},\ \bibinfo {pages} {120503} (\bibinfo {year} {2022})}\BibitemShut {NoStop}%
\bibitem [{\citenamefont {Ender}\ \emph {et~al.}(2022)\citenamefont {Ender}, \citenamefont {Messinger}, \citenamefont {Fellner}, \citenamefont {Dlaska},\ and\ \citenamefont {Lechner}}]{Ender_22}%
  \BibitemOpen
  \bibfield  {author} {\bibinfo {author} {\bibfnamefont {K.}~\bibnamefont {Ender}}, \bibinfo {author} {\bibfnamefont {A.}~\bibnamefont {Messinger}}, \bibinfo {author} {\bibfnamefont {M.}~\bibnamefont {Fellner}}, \bibinfo {author} {\bibfnamefont {C.}~\bibnamefont {Dlaska}},\ and\ \bibinfo {author} {\bibfnamefont {W.}~\bibnamefont {Lechner}},\ }\bibfield  {title} {\bibinfo {title} {Modular parity quantum approximate optimization},\ }\href {https://doi.org/10.1103/PRXQuantum.3.030304} {\bibfield  {journal} {\bibinfo  {journal} {PRX Quantum}\ }\textbf {\bibinfo {volume} {3}},\ \bibinfo {pages} {030304} (\bibinfo {year} {2022})}\BibitemShut {NoStop}%
\end{thebibliography}%

\appendix


\section{Appendix A: Proof of \Cref{thm:MBQC_YZ_I_O}} \label{app:aux_results}

A computation in MBQC is defined by the positive measurement outcomes of a sequence of single qubit projective measurements on a graph state $\ket{G}$. To compute deterministically, a method for effectively correcting for the occurrence of negative outcomes is required. By restricting measurements to the XY-, XZ- or YZ-planes of the Bloch sphere, the positive and negative projections are related by conjugation via $Z$, $XZ$ and $X$ respectively. Since each of these operators features in a tensor factor of some element of the stabiliser for $\ket{G}$, it is possible to correct for a negative outcome by conditionally ``completing'' that stabiliser element. 

The following definition, due to Browne et al. \cite{browne2007generalized}, outlines the criteria that a graph $G$, choice of input and output sets $I$ and $O$, and assignment of measurement planes to qubits must satisfy so that any measurement is correctable.

\begin{Definition} \label{def:gflow} Let $G = (V, E)$ be a graph, $I$ and $O$ be input and output subsets of $V$ respectively, and $\omega: V \setminus O \rightarrow \{XY, XZ, YZ\}$ be a map assigning measurement planes to qubits (the superscript $c$ denotes set complement). The tuple $(G,I,O, \omega)$ has \textbf{gflow} if there exists a map $g: V \setminus O \rightarrow \mathcal{P}(V \setminus I)$, where $\mathcal{P}$ denotes the powerset, and a partial order over $V$ such that the following hold for all $v \in V \setminus O$:
\begin{enumerate}
	\item if $v' \in g(v)$ and $v' \neq v$, then $v < v'$;
	\item if $v' \in \Odd(g(v))$ and $v' \neq v$, then $v < v'$;
	\item if $\omega(v) = XY$, then $v \notin g(v)$ and $v \in \Odd(g(v))$;
	\item if $\omega(v) = XZ$, then $v \in g(v)$ and $v \in \Odd(g(v))$;
	\item if $\omega(v) = YZ$, then $v \in g(v)$ and $v \notin \Odd(g(v))$;
\end{enumerate}
where $\Odd(K) := \{\tilde{v} \in V : |N_{\tilde{v}}^{G} \cap K| = 1 \bmod 2 \}$ for any $K \subseteq V$.
\end{Definition}

An intuition for the above definition is as follows. For each $v$, the sets $g(v)$ and $\Odd(g(v))$ specify the stabiliser element for the correction: $g(v)$ is the set of vertices that receive an $X$ and $\Odd(g(v))$ is the set of vertices that receive a $Z$. The first two conditions of the definition stipulate that the measurement corrections must occur in the future of the measurement, and the last three conditions enforce the correcting stabiliser element to have the correct tensor factor at each vertex corresponding to which plane that vertex was measured in. It was shown in Theorems $2$ and $3$ of \cite{browne2007generalized} that the gflow is a necessary and sufficient condition for deterministic MBQC. Accordingly, an equivalent statement to \Cref{thm:MBQC_YZ_I_O} in the main text is:

\begin{Theorem}\label{thm:rewrite} $(G, I, O, \omega \equiv YZ)$ with $|I| = |O|$ has gflow if and only if $G$ is bipartite with $I$ forming one partition.
\end{Theorem}

Before giving the proof of the above theorem, we require the following lemmas:

\begin{Lemma} \label{lem:I_equal_O} If $(G,I,O,\omega \equiv YZ)$ with $|I| = |O|$ has gflow, then $I = O$.
\end{Lemma}

\begin{proof} Let $(g, <)$ be a gflow for $(G,I,O,\omega)$ and suppose for a contradiction that $I \neq O$. Since $|I| = |O|$ this means that there exists some $v \in I$ such that $v \notin O$. Accordingly, $v$ is an element of the domain of the map $g$ but not the codomain so $v \notin g(v)$. Via criterion $5$ of \Cref{def:gflow}, this contradicts the assumption that $\omega(v) = YZ$.
\end{proof}

\begin{Lemma} \label{lem:po_restricted} Let $(G,I,O,\omega)$ with $I = O$ have gflow $(g,<)$ and let $\prec$ denote the partial order obtained by restricting $<$ to $V\setminus I$. Then any maximal element $v$ of $\prec$ must be such that $g(v) = \{v\}$ and hence $\omega(v) = YZ$.
\end{Lemma}

\begin{proof} Since $I = O$, $g$ maps from $V \setminus I$ to $\mathcal{P}(V \setminus I)$. If $v$ is maximal and $g(v) \neq \{v\}$ then there exists a $v' \in g(v) \setminus v$ (no $g(v)$ can be empty otherwise this would contradict criteria $3$, $4$ or $5$ of \Cref{def:gflow}). By criterion $1$ of \Cref{def:gflow}, this would mean that $v < v'$ and hence also $v \prec v'$, contradicting the maximality of $v$. Hence $g(v) = \{v\}$ and this must mean that $\omega(v) = YZ$ since otherwise this would contradict criteria $3$ or $4$.
\end{proof}

It is worth noting that, in the case where $I = O$, the restriction from $<$ to $\prec$ is not a significant one since every $v \in I$ is maximal in $<$.

\begin{Lemma} \label{lem:no_g_neighbours} Let $(G,I,O,\omega \equiv YZ)$ with $I = O$ have gflow. Then for every $u,v \in V\setminus I$ and $w \in g(u)$ and $x \in g(v)$, $w \notin N_{x}^{G}$.
\end{Lemma}

Before giving the proof, let us consider some consequences of this lemma. Taking $u = v = x$, we get that $N_{u}^{G} \cap g(u) = \emptyset$. Taking $u = v$ but $x \neq w$, we get that no two elements of $g(u)$ are connected by an edge. Taking $u \neq v$, we get that no element of $g(u)$ is connected to an element of $g(v)$. As a result, there are no edges between any elements in $\bigcup_{v \in V \setminus I} g(v)$, which is used in the proof of the theorem below.

\begin{proof}
Suppose $(g, <)$ is a valid gflow for $(G, I, O, \omega \equiv YZ)$ with $I=O$, which in particular means that $<$ is a valid partial order and criteria $1$, $2$ and $5$ in \Cref{def:gflow} are satisfied. Pursuant to \Cref{lem:po_restricted}, we make use of the partial order $\prec$ which is the restriction of $<$ to $V \setminus I$. 

Let $u,v \in V \setminus I$ and $w \in g(u)$ and $x \in g(v)$. Since $G$ is a simple graph, for $w = x$, $w \notin N_{x}^{G}$ by definition, so we consider $w \neq x$ forthwith. Suppose for a contradiction that $w \in N_{x}^{G}$. If $w \in \Odd(g(x))$ and $x \in \Odd(g(w))$, then a contradiction arises since criterion $2$ requires $w \prec x$ and $x \prec w$. We consider the following two cases in turn: (a) $w \notin \Odd(g(x))$ and $x \in \Odd(g(w))$ or $w \in \Odd(g(x))$ and $x \notin \Odd(g(w))$ and (b) $w \notin \Odd(g(x))$ and $x \notin \Odd(g(w))$.

\textbf{Case (a):} Since the two sub-cases are the same up to swapping $w$ and $x$, we consider without loss of generality the case where $w \notin \Odd(g(x))$ but $x \in \Odd(g(w))$. Since $w \neq x$, the latter gives that $ w \prec x$. Since $x \in N_{w}^{G}$ by assumption, there must exist a $x_{1} \in g(x)\setminus x$ such that $x_{1} \in N_{w}^{G}$ in order for $w \notin \Odd(g(x))$ to hold. If $w \in \Odd(g(x_{1}))$, then a contradiction arises since the criteria of gflow would require that $w \prec x \prec x_{1} \prec w$. If $w \notin \Odd(g(x_{1}))$, then there must exist a $x_{2} \in g(x_{1})\setminus x_{1}$ such that $x_{2} \in N_{w}^{G}$. By iterating the above reasoning, a sequence of vertices $x_{1}, x_{2}, ..., x_{l}$ is generated for which $x_{i+1} \in g(x_{i})\setminus x_{i}$, $x_{i+1} \in N_{w}^{G}$, and $w \notin \Odd(g(x_{i+1}))$ for each $i = 1, ..., l-1$. However, after finitely many steps, this sequence must terminate with some $x_{l+1}$ such that $w \in \Odd(g(x_{l+1}))$: in the worst case, this occurs when $x_{l+1}$ is a maximal element of $\prec$, which from \Cref{lem:po_restricted} means that $g(x_{l+1}) = \{x_{l+1}\}$ and thus $w \in \Odd(g(x_{l+1}))$. The contradiction arises since $w \prec x \prec x_{1} \prec  ... \prec x_{l} \prec x_{l+1} \prec w$.

\textbf{Case (b):} If $w \notin \Odd(g(x))$ and $x \notin \Odd(g(w))$ but $w \in N_{x}^{G}$, then there must exist $w_{1} \in g(w)\setminus w$ and $x_{1} \in g(x)\setminus x$ such that $x_{1} \in N_{w}^{G}$ and $w_{1} \in N_{x}^{G}$. If $w \in \Odd(g(x_{1}))$ and $x \in \Odd(g(w_{1}))$, then a contradiction arises since $w \prec w_{1} \prec x \prec x_{1} \prec w$. If $w \notin \Odd(g(x_{1}))$ and $x \in \Odd(g(w_{1}))$ or $w \in \Odd(g(x_{1}))$ and $x \notin \Odd(g(w_{1}))$, then we are in a similar situation to Case (a) above. By analogous reasoning, we obtain a sequence of vertices $x_{1} \prec x_{2} \prec ... \prec x_{l}$ such that there exists a $x_{l+1} \in g(x_{l})\setminus x_{l}$ for which $w \in \Odd(g(x_{l+1})$, producing a contradiction via $w \prec w_{1} \prec x \prec x_{1} \prec ...\prec x_{l+1} \prec w$.

If both $w \notin \Odd(g(x_{1}))$ and $x \notin \Odd(g(w_{1}))$ hold then we are in a scenario similar to that defining Case (b) to begin with. Accordingly, there must exist $w_{2} \in g(w_{1})\setminus w_{1}$ and $x_{2} \in g(x_{1})\setminus x_{1}$ such that $x_{2} \in N_{w}^{G}$ and $w_{2} \in N_{x}^{G}$. By iterating the above splitting into cases and the corresponding reasoning, we obtain sequences of vertices $w \prec w_{1} \prec ...\prec w_{l}$ and $x \prec x_{1} \prec  ... \prec x_{l}$ which must terminate (in the worst case) with either $g(w_{l}) \setminus w_{l}$ or $g(x_{l}) \setminus x_{l}$ containing a maximal element $w_{l+1}$ or $x_{l+1}$ respectively, and hence $w \in \Odd(g(x_{l+1}))$ or $x \in \Odd(g(w_{l+1}))$. Without loss of generality, suppose that $x \in \Odd(g(w_{l+1}))$. If $w \in \Odd(g(x_{l+1}))$, then we obtain a contradiction via $w \prec w_{1} \prec ... \prec w_{l+1} \prec x \prec x_{1} \prec ...\prec x_{l+1} \prec w$. If $ w \notin \Odd(g(x_{l+1}))$, then we are again in a scenario similar to Case (a), which by analogous reasoning also terminates in a contradiction.
\end{proof}
The proof of \Cref{thm:rewrite} proceeds as follows:
\begin{proof} Suppose $G$ is bipartite and denote one partition by $I$. Consider $(g, <)$ where $g$ is defined by $v \mapsto \{v\}$ for all $v \in  V\setminus I$ and (ii) $<$ is the coarsest partial order such that the vertices in $V \setminus I$ precede those in $I$. That criteria $1$, $2$ and $5$ in the definition of gflow are satisfied for all $v \in V\setminus I$ can be readily verified, hence $(g,<)$ is a valid gflow for $(G,I,O=I,\omega \equiv YZ)$. Clearly, with this choice of $O$, the requirement that $|I| = |O|$ is satisfied.

For the opposite direction, suppose $(G,I,O,\omega)$ with $\omega \equiv YZ$ and $|I| = |O|$ has gflow $(g, <)$. From \Cref{lem:I_equal_O}, we know that $I = O$. From \Cref{lem:no_g_neighbours}, we know that for each $u, v \in V \setminus I$, there are no edges between elements of $g(u)$ and also no edges between elements of $g(u)$ and $g(v)$. Consequently, there are no edges between elements of $g(V \setminus I) := \bigcup_{v \in V \setminus I} g(v)$. Accordingly, we can partition $G$ into those vertices in $g(V\setminus I) \equiv V\setminus I$ and those in $I$.
\end{proof}

\end{document}